\documentclass[journal]{IEEEtran}

\usepackage{xcolor,soul,framed}
\colorlet{shadecolor}{yellow}
\usepackage[pdftex]{graphicx}
\graphicspath{{../pdf/}{../jpeg/}}
\DeclareGraphicsExtensions{.pdf,.jpeg,.png}

\usepackage{cite}

\usepackage[cmex10]{amsmath}
\usepackage{amsfonts,amssymb,amsthm}
\newtheorem{theorem}{Theorem}
\theoremstyle{definition}
\newtheorem{example}{Example}

\usepackage[ruled,vlined]{algorithm2e}
\SetKwInput{KwInput}{Input}

\usepackage{array}
\usepackage{mdwmath}
\usepackage{mdwtab}
\usepackage{url}

\begin{document}
\bstctlcite{IEEEexample:BSTcontrol}
    \title{Versatile Time-Frequency Representations Realized by Convex Penalty on Magnitude Spectrogram}
    \author{Keidai~Arai,~Koki~Yamada,~\IEEEmembership{Member,~IEEE},~and~Kohei~Yatabe,~\IEEEmembership{Member,~IEEE}
    \thanks{Manuscript received May XX, 2023; revised XXXXXX XX, 2023; accepted XXXXXX XX, 2023. Date of publication XXXXXX XX, 2023; date of current version XXXXXX XX, 2023. The associate editor coordinating the review of this manuscript and approving it for publication was Dr. XXXXXX XXXXXX. (Corresponding author: Keidai Arai.)}%
    \thanks{The authors are with Tokyo University of Agriculture and Technology, Tokyo 184-8588, Japan (e-mail: email@email.com; email@email.com; yatabe@go.tuat.ac.jp).}%
    \thanks{Digital Object Identifier 10.1109/LSP.2023.XXXXXXX}}

\markboth{IEEE SIGNAL PROCESSING LETTER}{Arai \MakeLowercase{\textit{et al.}}: Versatile Time-Frequency Representations}

\maketitle

\begin{abstract}
Sparse time-frequency (T-F) representations have been an important research topic for more than several decades.
Among them, optimization-based methods (in particular, extensions of basis pursuit) allow us to design the representations through objective functions.
Since acoustic signal processing utilizes models of spectrogram, the flexibility of optimization-based T-F representations is helpful for adjusting the representation for each application.
However, acoustic applications often require models of \textit{magnitude} of T-F representations obtained by discrete Gabor transform (DGT).
Adjusting a T-F representation to such a magnitude model (e.g., smoothness of magnitude of DGT coefficients) results in a non-convex optimization problem that is difficult to solve.
In this paper, instead of tackling difficult non-convex problems, we propose a convex optimization-based framework that realizes a T-F representation whose magnitude has characteristics specified by the user.
We analyzed the properties of the proposed method and provide numerical examples of sparse T-F representations having, e.g., low-rank or smooth magnitude, which have not been realized before.
\end{abstract}


\IEEEpeerreviewmaketitle

\section{Introduction}

\IEEEPARstart{T}{ime-frequency} (T-F) analysis is an essential tool in science and engineering \cite{DGT, DGT2}. 
The topic of this paper can include any complex-valued T-F analysis, but we focus on the short-time Fourier transform (STFT), or discrete Gabor transform (DGT), for brevity.
Over several decades, sparse T-F analysis has been an important research topic for breaking the barrier of the uncertainty principle.
For example, reassignment and synchrosqueezing have been applied to STFT/DGT for computing sharper spectrograms \cite{reassignment, synchrosqueezing, synchrosqueezing2, reassignment_synchrosqueezing, reassignment_synchrosqueezing_problem}.

Optimization-based sparse T-F analysis has offered flexibility for designing T-F representation.
Thanks to the redundancy of DGT, T-F representation can be customized by formulating  an optimization problem and solving it.
The obtained representation has properties imposed by the penalty function defined in the optimization problem.
For example, the $\ell_1$-norm used in the basis pursuit problem enhances sparsity \cite{basis_pursuit}, and the mixed-norm promotes structured sparsity determined by its definition \cite{KowalskiMixedNorm2009}.
There are many other penalty functions that can be used for designing T-F representations \cite{basis_pursuit_tf, basis_pursuit_tf2, enhancing_sparsity, IRLS, elastic_net, L1-2_compress, atomic_norm, block_sparsity2, block_sparsity3, block_sparsity, block_sparsity4, LOP, social_sparsity, Bach2012Sparse}.
By choosing an appropriate penalty function, one can adjust a T-F representation to its application.

However, practical applications in acoustics often assume \textit{magnitude} of DGT coefficients to have specific characteristics.
For example, low-rankness of spectrogram is often assumed, which resulted in popularity of applying nonnegative matrix factorization (NMF) to spectrogram~\cite{NMF_Bayesian, NMF_multichannel, ILRMA}.
Smoothness of spectrogram has also be utilized~\cite{OnoHPSS, SmoothSpecNMF, PhaseAwareHPSS}.
To impose these properties on T-F representations, naive formulation requires a penalty function that handles magnitude of complex numbers, which easily results in a difficult non-convex optimization problem.
For example, spectrogram smoothness requires to penalize difference of magnitude, which leads to a composition of non-linear transform (i.e., absolute value), linear transform (i.e., difference operator) and a norm.
Such composition including absolute value usually results in a non-convex penalty function.
This difficulty has obstructed optimization-based T-F analysis to be applied in practice.

To resolve this issue, we propose a convex optimization-based framework that can penalize magnitude of T-F representations.
We introduce a nonnegative auxiliary variable related to magnitude of the T-F representation, and a penalty function is applied to it.
This auxiliary variable is combined with the T-F representation using a perspective function~\cite{perspective, perspective2, perspective3, perspective4} so that the overall optimization problem is convex whenever the penalty function is convex.
Our contributions in this paper can be summarized as follows:
(i) formulating a novel convex optimization problem for sparse T-F analysis; (ii) discussing the properties of the proposed optimization problem; (iii) deriving a primal-dual algorithm; and (iv) providing numerical examples.
The proposed framework realizes some completely new T-F representations that have not been available before.

\textbf{Notations.} \;
$\mathbb{N}$, $\mathbb{R}$, $\mathbb{R}_+$ and $\mathbb{C}$ denote the sets of all positive integers, real numbers, nonnegative numbers, and complex numbers, respectively.
$\overline{(\cdot)}$, $(\cdot)^\mathrm{T}$, $(\cdot)^\mathrm{H}$, $|\cdot|$ and $\odot$ denote the complex conjugate, transpose, Hermitian transpose, entry-wise absolute value, and entry-wise multiplication, respectively.
The $\ell_1$- and $\ell_2$-norms are $\|\boldsymbol{x}\|_1 = \sum_{n=1}^N|x_n|$ and $\|\boldsymbol{x}\|_2 = \sqrt{\boldsymbol{x}^\mathrm{H}\boldsymbol{x}}$, respectively.
The nuclear norm $\|\cdot\|_\ast$ is the $\ell_1$-norm of singular values, and $\|\cdot\|_\mathrm{op}$ is the operator norm.
The set of all proper lower semicontinuous convex functions is denoted by $\Gamma_0(\mathbb{R}^N)$.
The proximity operator of a function $f\in\Gamma_0(\mathbb{R}^N)$ is denoted by $\mathrm{prox}_{f}(\boldsymbol{x}) = \arg\min_{\boldsymbol{\xi}\in\mathbb{R}^N}\,f(\boldsymbol{\xi}) + {\textstyle\frac{1}{2}}\left\|\boldsymbol{x}-\boldsymbol{\xi}\right\|_2^2$.

\section{Preliminaries}

\subsection{Discrete Gabor Transform and Basis Pursuit}

Let DGT of $\boldsymbol{d}\in\mathbb{C}^L$ with respect to $\boldsymbol{w}\in\mathbb{C}^L$ be defined as
\begin{align}\label{eq:dgt}
    x_{m,n} = \displaystyle\sum_{l=0}^{L-1} d_{l}\,\overline{w_{l-an}} \,\mathrm{e}^{-2\pi \mathrm{i} ml/M},
\end{align}
where $n = 0, \dotsc, N-1$ and $m = 0, \dotsc, M-1$ are the time and frequency indices, respectively, $N = L/a$ and $M \in \mathbb{N}$ are the numbers of time frames and frequency bins, respectively, and $a \in \mathbb{N}$ is the time-shifting width.
The signal length $L$ is assumed to satisfy $N = L/a \in \mathbb{N}$ and $MN > L$.
Eq.~\eqref{eq:dgt} can be shortly written using the matrix $\boldsymbol{G}_w \in \mathbb{C}^{MN \times L}$ as
\begin{align}\label{eq:dgt_matrix}
\boldsymbol{x} = \boldsymbol{G}_w\boldsymbol{d},
\end{align}
where $(G_w)_{m+nM,l} = \overline{w_{l-an}} \,\mathrm{e}^{-2\pi \mathrm{i} ml/M}$.
If $\boldsymbol{G}_w^\mathrm{H}\boldsymbol{G}_w$ is invertible, there exists the canonical dual window of $\boldsymbol{w}$ given by
\begin{equation}\label{eq:dual_window}
\boldsymbol{\gamma}^\star = (\boldsymbol{G}_w^\mathrm{H} \boldsymbol{G}_w)^{-1} \boldsymbol{w},
\end{equation}
which admits the following important identity:
\begin{equation}\label{eq:idgt}
\boldsymbol{G}_{\gamma^\star}^\mathrm{H} \boldsymbol{G}_w = \boldsymbol{G}_w^\mathrm{H} \boldsymbol{G}_{\gamma^\star} = \boldsymbol{I},
\end{equation}
where $\boldsymbol{I}$ denotes the identity matrix.

As $MN>L$, redundancy of DGT can be used for adjusting a T-F representation.
For example, solving the following basis pursuit problem gives a sparse T-F representation \cite{basis_pursuit, basis_pursuit_tf}:
\begin{align}\label{eq:basis_pursuit}
    \min_{\boldsymbol{x}\in\mathbb{C}^{MN}} \,\left\|\boldsymbol{x}\right\|_1 \,\; \mathrm{s.t.} \;\; \boldsymbol{G}_{\gamma^\star}^\mathrm{H}\boldsymbol{x} = \boldsymbol{d}.
\end{align}
The $\ell_1$-norm is the most standard convex penalty function for inducing sparsity.
Using other penalty functions in place of the $\ell_1$-norm results in different T-F representations.

\subsection{Structured Penalty Functions}

Since the $\ell_1$-norm is not the best choice, many other penalty functions have been proposed such as non-convex penalty functions \cite{GISA, Chartrand2014ShrinkageMap, IvanBayram2014MaximallySparse, Ivan2017Convex} and structured penalty functions \cite{block_sparsity, block_sparsity2, block_sparsity3, block_sparsity4, LOP, social_sparsity, Bach2012Sparse}.
We do not consider non-convex penalty functions in this paper because they result in non-convex optimization problems which are difficult to solve globally.
Structured penalty functions have flexibility for incorporating some knowledge on structure of data (e.g., grouped or tree structure).
However, they cannot handle some structures, e.g., those determined by difference between the magnitude of T-F bins.
Moreover, it is difficult to handle non-local structures, and hence most structured penalty functions focus on local relationship.
The proposed framework aims to overcome these limitations.

The most important for interpreting our proposal but often unnoticed alternative for structured optimization is weighted norms \cite{enhancing_sparsity,IRLS}.
The weighted $\ell_2$- and $\ell_1$-norms can be defined as $\sqrt{\boldsymbol{x}^\mathrm{H}\boldsymbol{\Sigma}^{-1}\boldsymbol{x}}$ and $\|\boldsymbol{\Sigma}^{-1}\boldsymbol{x}\|_1$, respectively, where the weight is represented by $\boldsymbol{\sigma}>\boldsymbol{0}$ and $\boldsymbol{\Sigma} = \mathrm{diag(\boldsymbol{\sigma})}$.
If this weight is specifically designed according to prior knowledge, the weighted norms induce the property determined by the weight.
For example, if one knows which entries of the solution to be small, then setting large weights to those entries results in a solution satisfying the prior knowledge.

\section{Proposed Method}

In this section, we propose a novel framework for realizing sparse T-F representations having desired magnitude.
As the proposed method relies on a perspective function~\cite{perspective, perspective2, perspective3, perspective4}, it is briefly reviewed before introducing the proposed method.

\subsection{Perspective Function for Optimizing Weighted Norm}

The proposed method relies on the following convex function $\varphi$ defined for a pair $(\boldsymbol{x}, \boldsymbol{\sigma})\in\mathbb{C}^{MN}\times\mathbb{R}^{MN}$ as follows:
\begin{align}\label{eq:sum_perspective}
    \varphi(\boldsymbol{x}, \boldsymbol{\sigma}) = \sum_{k=1}^{MN}\phi(x_k, \sigma_k),
\end{align}
where $\phi : \mathbb{C}\times\mathbb{R} \rightarrow \mathbb{R}_+ \cup \{\infty\}$ is given by
\begin{equation}\label{eq:perspective}
\phi(x_k, \sigma_k) =
\left\{
    \begin{array}{cl}
    \frac{|x_k|^2}{2\sigma_k} + \frac{\sigma_k}{2} & (\sigma_k > 0),\\
    0 & (x_k = 0 \:\:\mathrm{and}\:\: \sigma_k = 0),\\
    \infty & (\mathrm{otherwise}).
    \end{array}
\right.
\end{equation}
This is the perspective function of $(|\cdot|^2/2) + (1/2)$ and hence a proper lower semicontinuous convex function \cite{perspective, perspective3}.

If $\boldsymbol{\sigma}>\boldsymbol{0}$, then $2\,\varphi$ can be viewed as the squared weighted $\ell_2$-norm for $\boldsymbol{x}$, i.e., $\boldsymbol{x}^\mathrm{H}\boldsymbol{\Sigma}^{-1}\boldsymbol{x} + \|\boldsymbol{\sigma}\|_1$ with $\boldsymbol{\Sigma}=\mathrm{diag}(\boldsymbol{\sigma})$.
Minimizing $\varphi$ can simultaneously penalize $\boldsymbol{x}$ and optimize $\boldsymbol{\sigma}$, and hence $\varphi$ can be interpreted as a squared weighted $\ell_2$-norm with an adaptive weight.
By properly modifying the weight $\boldsymbol{\sigma}$, a desired structure can be imposed on $\boldsymbol{x}$ through $\varphi$.

\subsection{Convex Penalty on Magnitude of DGT Coefficients} \label{sec:proveMagnitude}

To modify the weight $\boldsymbol{\sigma}$, we introduce a penalty function $\Psi\in\Gamma_0(\mathbb{R}^{MN})$.
By replacing the $\ell_1$-norm of the basis pursuit problem in Eq.~\eqref{eq:basis_pursuit} with $\varphi$ and $\Psi$, we obtain the proposed convex optimization problem for sparse T-F representation:
\begin{align}\label{eq:proposed_formulation}
    \min_{(\boldsymbol{x}, \boldsymbol{\sigma}) \in \mathbb{C}^{MN} \times \mathbb{R}^{MN}} \,\varphi(\boldsymbol{x}, \boldsymbol{\sigma}) + \Psi(\boldsymbol{\sigma}) \;\; \mathrm{s.t.} \;\; \boldsymbol{G}_{\gamma^\star}^\mathrm{H}\boldsymbol{x} = \boldsymbol{d}.
\end{align}
This formulation allows us to design $\Psi$ that imposes a structure on the weight $\boldsymbol{\sigma}$, which is transferred to the DGT coefficients $\boldsymbol{x}$ via the weighted $\ell_2$-norm inside $\varphi$.

At first glance, it might be unclear how $\Psi$ affects $\boldsymbol{x}$ because of the indirect formulation.
Here, we show that $\boldsymbol{\sigma}$ is actually related to the magnitude of DGT coefficients $|\boldsymbol{x}|$.
According to the following result (and examples in Section~\ref{sec:exp}), we regard $\Psi(\boldsymbol{\sigma})$ as an indirect penalty function for $|\boldsymbol{x}|$.
\begin{theorem}
    For each $\boldsymbol{x}$, let $\boldsymbol{\sigma}^\star_{\!\boldsymbol{x}}$ be a minimizer of $\varphi(\boldsymbol{x}, \boldsymbol{\sigma}) + \Psi(\boldsymbol{\sigma})$.
    If $\Psi = 0$, then $\boldsymbol{\sigma}^\star_{\!\boldsymbol{x}} = |\boldsymbol{x}|$.
    If $\Psi \neq 0$ and $\boldsymbol{\sigma}^\star_{\!\boldsymbol{x}} = |\boldsymbol{x}|$, then $\boldsymbol{\sigma}^\star_{\!\boldsymbol{x}}$ minimizes $\Psi(\boldsymbol{\sigma})$.
\end{theorem}
\begin{proof}
    Since $\boldsymbol{\sigma}^\star_{\!\boldsymbol{x}}$ minimizes $\varphi(\boldsymbol{x} , \boldsymbol{\sigma}) + \Psi(\boldsymbol{\sigma})$ for a fixed $\boldsymbol{x}$, it satisfies the following optimality condition:
    \begin{equation}\label{eq:kkt_sigma}
        \boldsymbol{0} \in \frac{|\boldsymbol{x}|^2}{2} \odot \left(-\frac{1}{(\boldsymbol{\sigma}^\star_{\!\boldsymbol{x}})^2}\right) + \frac{1}{2} + \partial\Psi(\boldsymbol{\sigma}^\star_{\!\boldsymbol{x}}).
    \end{equation}
    If $\Psi=0$, then $\partial\Psi = \{\boldsymbol{0}\}$ and hence $\boldsymbol{\sigma}^\star_{\!\boldsymbol{x}} = |\boldsymbol{x}|$.
    If $\Psi\neq0$, substituting $\boldsymbol{\sigma}^\star_{\!\boldsymbol{x}} = |\boldsymbol{x}|$ into Eq.~\eqref{eq:kkt_sigma} gives $\boldsymbol{0} \in \partial\Psi(|\boldsymbol{x}|)$, and hence $\boldsymbol{\sigma}^\star_{\!\boldsymbol{x}} = |\boldsymbol{x}|$ minimizes $\Psi(\boldsymbol{\sigma})$.
\end{proof}

\subsection{Primal-Dual Algorithm for the Proposed Framework}

Consider the following specific form of Problem \eqref{eq:proposed_formulation}:
\begin{align}\label{eq:proposed_reformulation}
    \min_{(\boldsymbol{x}, \boldsymbol{\sigma}) \in \mathbb{C}^{MN} \times \mathbb{R}^{MN}} \varphi(\boldsymbol{x}, \boldsymbol{\sigma}) + \lambda\psi(\boldsymbol{B\sigma}) + \iota_\mathcal{C}(\boldsymbol{x}),
\end{align}
where $\boldsymbol{B}\in\mathbb{C}^{J\times MN}$, $\psi \in \Gamma_0(\mathbb{R}^J)$, $\lambda > 0$, $\iota_\mathcal{C}(\boldsymbol{x})$ is the indicator function of $\mathcal{C}$ (i.e., $\iota_\mathcal{C}(\boldsymbol{x})=0$ if $\boldsymbol{x}\in\mathcal{C}$, and $\iota_\mathcal{C}(\boldsymbol{x})=\infty$ otherwise), and $\mathcal{C} = \{\boldsymbol{x} \in \mathbb{C}^{MN} \mid \boldsymbol{G}_{\gamma^\star}^\mathrm{H}\boldsymbol{x} = \boldsymbol{d}\}$.
Let us provide some examples of the penalty function $\psi\circ\boldsymbol{B}$.
\begin{example}\label{example}
    Structures induced by $\psi\circ\boldsymbol{B}$ in Problem \eqref{eq:proposed_reformulation}.
    \begin{enumerate}
    \renewcommand{\theenumi}{(\roman{enumi}}
        \item \textbf{Sparsity:}
        Setting $\psi = \|\cdot\|_1$ and $\boldsymbol{B}=\boldsymbol{I}$ induces sparsity of $|\boldsymbol{x}|$.
        Note that $\psi = 0$ also induces sparsity because $\|\boldsymbol{\sigma}\|_1$ is included in the definition of $\varphi(\boldsymbol{x},\boldsymbol{\sigma})$ in Eq.~\eqref{eq:sum_perspective}.
        \item \textbf{Low-rankness:}
        Setting $\psi = \|\cdot\|_\ast$ and $\boldsymbol{B} = \boldsymbol{I}$ induces low-rankness of $|\boldsymbol{x}|$ \cite{nuclear_norm}, where the nuclear norm treats $\boldsymbol{\sigma}\in\mathbb{R}^{MN}$ as an $M\times N$ matrix.
        Note that the nuclear norm cannot be directly applied to $\boldsymbol{x}$ because $\boldsymbol{x}$ cannot be low-rank in the complex-valued sense \cite{Masuyama2019LowRank}.
        
        \item \textbf{Total variation:}
        Setting $\psi = \|\cdot\|_{2,1}$ and $\boldsymbol{B} = \boldsymbol{D}$ induces smoothness of $|\boldsymbol{x}|$ \cite{total_variation}, where the difference matrix $\boldsymbol{D}$ approximates the gradient at each $M\times N$ entry, and $\|\cdot\|_{2,1}$ penalizes sum of magnitude of the gradients.
        \item \textbf{Harmonic enhancement:} Setting $\psi = \|\cdot\|_{2,1}$ and $\boldsymbol{B} = \boldsymbol{C}\boldsymbol{D}$ enhances the harmonic structure of $|\boldsymbol{x}|$ \cite{harmonic_enhancement}, where $\boldsymbol{C}$ denotes the discrete cosine transform (DCT) along the frequency axis. Making DCT coefficients sparse emphasizes periodic patterns of $|\boldsymbol{x}|$.
    \end{enumerate}
\end{example}

Let $\boldsymbol{L} = [ [\boldsymbol{I}\;\boldsymbol{O}]^\mathrm{T},[\boldsymbol{O}\;\boldsymbol{B}]^\mathrm{T} ]$, $\boldsymbol{y} =[ \boldsymbol{x}^\mathrm{T}, \boldsymbol{\sigma}^\mathrm{T} ]^\mathrm{T}$, $f(\boldsymbol{y}) = \varphi(\boldsymbol{x}, \boldsymbol{\sigma})$, and $g(\boldsymbol{Ly}) = \iota_\mathcal{C}(\boldsymbol{x}) + \lambda\psi(\boldsymbol{B\sigma})$.
Then, Problem \eqref{eq:proposed_reformulation} can be rewritten as $\min_{\boldsymbol{y}}\, f(\boldsymbol{y}) + g(\boldsymbol{Ly})$.
Directly applying the well-known Chambolle--Pock algorithm \cite{ChambollePock,relax_them_all} to this problem provides \mbox{\textbf{Algorithm \ref{alg:prop}}}, where the two proximity operators, $\mathrm{prox}_{\tau\varphi}$ and $P_\mathcal{C}$, are given as follows.

\begin{figure*}
    \centering
    \vspace{-15pt}
    \includegraphics[scale=0.785]{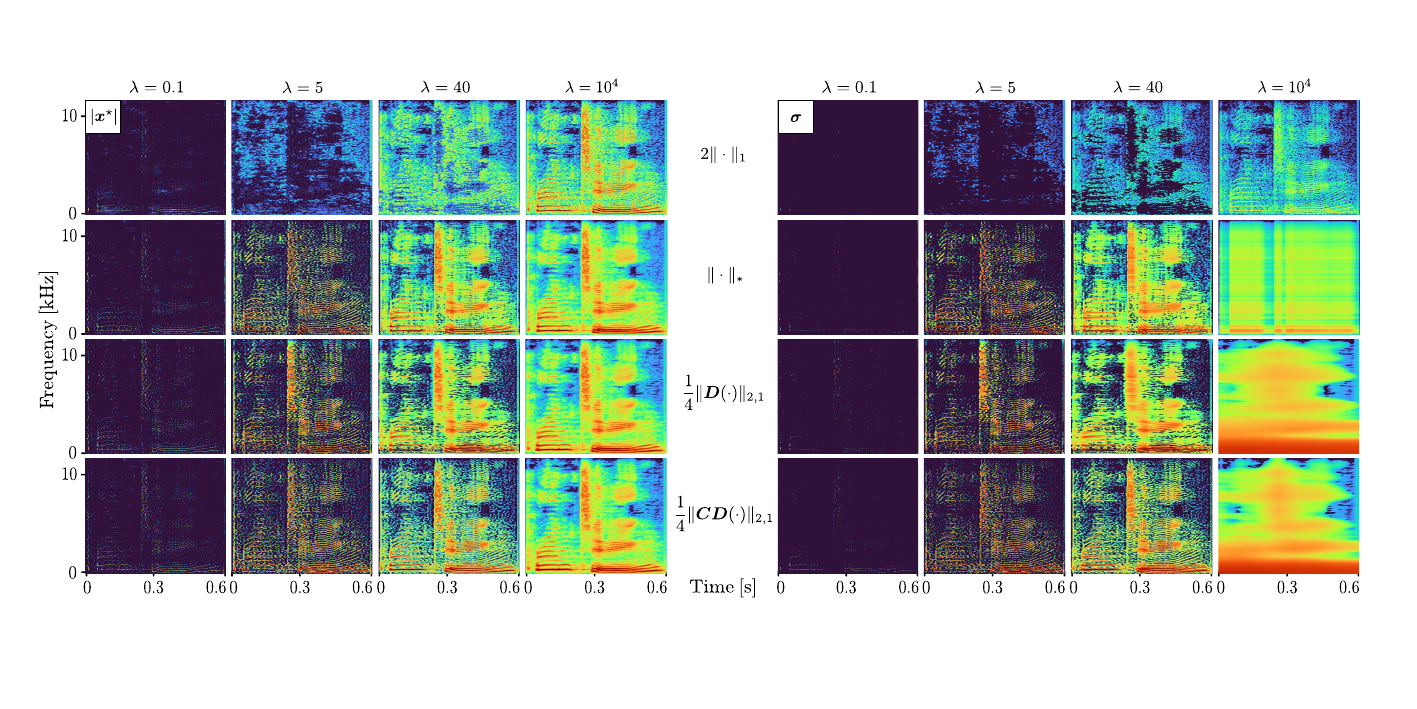}
    \vspace{-55pt}
    \caption{Obtained T-F representations $|\boldsymbol{x}^{\star\!}|$ of a speech signal (left) and corresponding auxiliary variables $\boldsymbol{\sigma}$ (right).
    All figures are illustrated by taking $20\log_{10}(\cdot)$, and the color range is 100 dB.
    Each row (from top to bottom) corresponds to the penalty $\psi(\boldsymbol{B}(\cdot)) = 2\|\cdot\|_1$, $\|\cdot\|_\ast$, $\frac{1}{4}\|\boldsymbol{D}(\cdot)\|_{2,1}$, and $\frac{1}{4}\|\boldsymbol{CD}(\cdot)\|_{2,1}$, respectively, where the coefficients are chosen for better visibility.
    Each column corresponds to different parameter $\lambda = 0.1, 5, 40, 10^4$.
    } 
    \label{fig:x_sigma_summary}
\end{figure*}

Owing to the separability of $\varphi$ in Eq.~\eqref{eq:sum_perspective}, the proximity operator of $\tau\varphi$ can be computed entry-wise \cite{SeparableSum},
\begin{align}\label{eq:prox_tau_varphi}
    \mathrm{prox}_{\tau\varphi}(\boldsymbol{x}, \boldsymbol{\sigma}) = (\mathrm{prox}_{\tau\phi}(x_k, \sigma_k))_{k=1}^{MN}.
\end{align}
The proximity operator for each entry is given as~\cite{perspective3}
\begin{align}\label{eq:prox_tau_phi}
    &\mathrm{prox}_{\tau\phi}(x_k, \sigma_k) \vspace{1mm}\nonumber\\
    &= \left\{\!\!
    \begin{array}{cl}
        (0,0) & (2\tau\sigma_k + |x_k|^2 \leq \tau^2),\\
        (0, \sigma_k - \frac{\tau}{2}) & (x_k\! = 0 \;\,\mathrm{and}\;\, 2\sigma_k\! > \tau),\\
        (x_k - \tau s \frac{x_k}{|x_k|}, \sigma_k + \tau\frac{s^2 - 1}{2}) & (\mathrm{otherwise}),
    \end{array}
    \right.
\end{align}
where $s > 0$ is the unique positive root of the cubic equation $s^3 + (\frac{2}{\tau}\sigma_k + 1)s - \frac{2}{\tau}|x_k| = 0$.
This cubic equation can be solved using Cardano's formula as follows:
\begin{align}\label{eq:s_formula}
    s = \left\{
    \begin{array}{ll}
        \sqrt[3]{-\frac{q}{2} + \sqrt{-r}} + \sqrt[3]{-\frac{q}{2} - \sqrt{-r}} & (r < 0),\\
        2\:\sqrt[3]{-\frac{q}{2}} & (r = 0),\\
        2\:\sqrt[3]{\sqrt{\frac{q^2}{4} + r}}\,\cos{\left(\frac{\arctan{(-2\sqrt{r}/q)}}{3}\right)} & (r > 0),
    \end{array}
    \right.
\end{align}
where $p=\frac{2}{\tau}\sigma_k + 1$, $q = -\frac{2}{\tau}|x_k|$, $r = -\frac{q^2}{4}-\frac{p^3}{27}$, and $\sqrt[3]{\cdot}$ is the real cubic root.
The projection onto $\mathcal{C}$ can be computed as
\begin{align}\label{eq:projection_c}
    P_\mathcal{C}(\boldsymbol{x}) &= \boldsymbol{x} - \boldsymbol{G}_{\gamma^\star}(\boldsymbol{G}_{\gamma^\star}^\mathrm{H}\boldsymbol{G}_{\gamma^\star})^{-1}(\boldsymbol{G}_{\gamma^\star}^\mathrm{H}\boldsymbol{x} - \boldsymbol{d})\nonumber\\
    &= \boldsymbol{x} - \boldsymbol{G}_{w}(\boldsymbol{G}_{\gamma^\star}^\mathrm{H}\boldsymbol{x} - \boldsymbol{d}).
\end{align}

The proximity operator of $(\lambda/\mu)\psi$ depends on the choice of the penalty function $\psi$.
The sequence $(\boldsymbol{x}^{[i]},\boldsymbol{\sigma}^{[i]})_{i\in\mathbb{N}}$ generated by the algorithm converges to a globally optimal solution of Problem \eqref{eq:proposed_reformulation} if the following conditions are satisfied \cite{relax_them_all}: $\tau\mu\|\boldsymbol{L}\|_\mathrm{op}^2 \leq 1$ and $\sum_{i \in \mathbb{N}} \rho^{[i]}(2 - \rho^{[i]}) = \infty$.

\begin{algorithm}[t]
\label{alg:prop}
\DontPrintSemicolon
  \KwInput{$\tau > 0,\: \mu > 0,\: \rho^{[i]} \in (0,2)\:(i=0, 1, 2, \dotsc)$,~and \mbox{$\boldsymbol{x}^{[0]} \in \mathbb{C}^{MN}, \boldsymbol{\sigma}^{[0]} \in \mathbb{R}^{MN}, \boldsymbol{u}^{[0]} \in \mathbb{C}^{MN}, \boldsymbol{v}^{[0]} \in \mathbb{C}^{J}$}}
  \For{$i=0, 1, 2, \dotsc$}
  {
    $\tilde{\boldsymbol{x}}^{[i+\frac{1}{2}]} = \boldsymbol{x}^{[i]} - \tau \boldsymbol{u}^{[i]}$\\
    \vspace{0.3mm}
    $\tilde{\boldsymbol{\sigma}}^{[i+\frac{1}{2}]} = \boldsymbol{\sigma}^{[i]} - \tau\boldsymbol{B}^\mathrm{H}\boldsymbol{v}^{[i]}$\\
    \vspace{0.3mm}
    $(\boldsymbol{x}^{[i+\frac{1}{2}]}, \boldsymbol{\sigma}^{[i+\frac{1}{2}]}) = \mathrm{prox}_{\tau\varphi}(\tilde{\boldsymbol{x}}^{[i+\frac{1}{2}]}, \tilde{\boldsymbol{\sigma}}^{[i+\frac{1}{2}]})$\\
    \vspace{0.3mm}
    $\tilde{\boldsymbol{u}}^{[i+\frac{1}{2}]} = \boldsymbol{u}^{[i]} + \mu(2\boldsymbol{x}^{[i+\frac{1}{2}]} - \boldsymbol{x}^{[i]})$\\
    \vspace{0.3mm}
    $\boldsymbol{u}^{[i+\frac{1}{2}]} = \tilde{\boldsymbol{u}}^{[i+\frac{1}{2}]} - \mu P_\mathcal{C}(\tilde{\boldsymbol{u}}^{[i+\frac{1}{2}]} / \mu)$\\
    \vspace{0.3mm}
    $\tilde{\boldsymbol{v}}^{[i+\frac{1}{2}]} = \boldsymbol{v}^{[i]} + \mu\boldsymbol{B}(2\boldsymbol{\sigma}^{[i+\frac{1}{2}]} - \boldsymbol{\sigma}^{[i]})$\\
    \vspace{0.3mm}
    $\boldsymbol{v}^{[i+\frac{1}{2}]} = \tilde{\boldsymbol{v}}^{[i+\frac{1}{2}]} - \mu\:\mathrm{prox}_{(\lambda/\mu)\psi}(\tilde{\boldsymbol{v}}^{[i+\frac{1}{2}]} / \mu)$\\
    \vspace{0.3mm}
    $\boldsymbol{x}^{[i+1]} = \boldsymbol{x}^{[i]} + \rho^{[i]}(\boldsymbol{x}^{[i+\frac{1}{2}]} - \boldsymbol{x}^{[i]})$\\
    \vspace{0.3mm}
    $\boldsymbol{\sigma}^{[i+1]} = \boldsymbol{\sigma}^{[i]} + \rho^{[i]}(\boldsymbol{\sigma}^{[i+\frac{1}{2}]} - \boldsymbol{\sigma}^{[i]})$\\
    \vspace{0.3mm}
    $\boldsymbol{u}^{[i+1]} = \boldsymbol{u}^{[i]} + \rho^{[i]}(\boldsymbol{u}^{[i+\frac{1}{2}]} - \boldsymbol{u}^{[i]})$\\
    \vspace{0.3mm}
    $\boldsymbol{v}^{[i+1]} = \boldsymbol{v}^{[i]} + \rho^{[i]}(\boldsymbol{v}^{[i+\frac{1}{2}]} - \boldsymbol{v}^{[i]})$\\
  }
\caption{Solver for the proposed framework \eqref{eq:proposed_reformulation}}
\end{algorithm}

\subsection{Some Notes on the Property of the Proposed Method}

Although $\lambda$ in Problem \eqref{eq:proposed_reformulation} changes strength of penalty on $\boldsymbol{\sigma}$, larger $\lambda$ does not always imply stronger induction towards the structure induced by $\psi\circ\boldsymbol{B}$.
This is because the equality constraint restricts the solution to be in the feasible set, but the induced structure may not fit into the constraint.
For instance, (i) and (ii) of Example \ref{example} induces $\boldsymbol{\sigma} = \boldsymbol{0}$ when $\lambda$ is huge, but $\boldsymbol{x}$ cannot be $\boldsymbol{0}$ due to the constraint.
In that case, the structure induced by $\psi\circ\boldsymbol{B}$ may disappear because $\boldsymbol{\sigma}$ becomes more like constant which cannot impose any structure via the weighted norm of $\varphi$.
When $\boldsymbol{\sigma}$ is fixed to positive numbers, say $\tilde{\boldsymbol{\sigma}}>\boldsymbol{0}$, Problem \eqref{eq:proposed_reformulation} reduces to $\min_{\boldsymbol{x} \in \mathbb{C}^{MN}} \boldsymbol{x}^\mathrm{H}\tilde{\boldsymbol{\Sigma}}{}^{-1}\boldsymbol{x} + \iota_\mathcal{C}(\boldsymbol{x})$ \ $(\tilde{\boldsymbol{\Sigma}} = \mathrm{diag}(\tilde{\boldsymbol{\sigma}}))$ whose solution is $\boldsymbol{x}_{\tilde{\boldsymbol{\sigma}}}^\star = \tilde{\boldsymbol{\Sigma}}\boldsymbol{G}_{\gamma^\star}(\boldsymbol{G}_{\gamma^\star}^\mathrm{H}\tilde{\boldsymbol{\Sigma}}\boldsymbol{G}_{\gamma^\star})^{-1}\boldsymbol{d}$, and hence it becomes the minimum norm solution $\boldsymbol{G}_w\boldsymbol{d}$ when $\boldsymbol{\sigma}$ is fixed to a positive constant $c\mathbf{1}$ $(c>0)$.
Therefore, exceedingly large $\lambda$ may result in a solution $\boldsymbol{x}^\star$ close to $\boldsymbol{G}_w\boldsymbol{d}$.
However, $\boldsymbol{x}^\star = \boldsymbol{G}_w\boldsymbol{d}$ does not occur because the second term of $\varphi$ (i.e., $\|\boldsymbol{\sigma}\|_1/2$) induces sparsity regardless of the choice of $\lambda$ and $\psi\circ\boldsymbol{B}$.

\section{Numerical Examples}\label{sec:exp}
To illustrate the property of the proposed framework, some examples are shown here.
A speech signal was analyzed using the Hann window $(L = 2^9)$ with hop size $a = 2^6$ and frequency bins $M=2^{12}$.
For convergence, Algorithm \ref{alg:prop} was iterated $5000$ times using $\tau = 1 / 2$, $\mu = 1 / 5$, $\rho^{[i]} = 1.99$, $\boldsymbol{\sigma}^{[0]} = |\boldsymbol{G}_w\boldsymbol{d}|$, $\boldsymbol{u}^{[0]} = \boldsymbol{0}$, and $\boldsymbol{v}^{[0]} = \boldsymbol{0}$.
For the penalty function $\psi\circ\boldsymbol{B}$, those listed in Example~\ref{example} were used.

\begin{figure}[!t]
    \centering
    \includegraphics[scale=0.45]{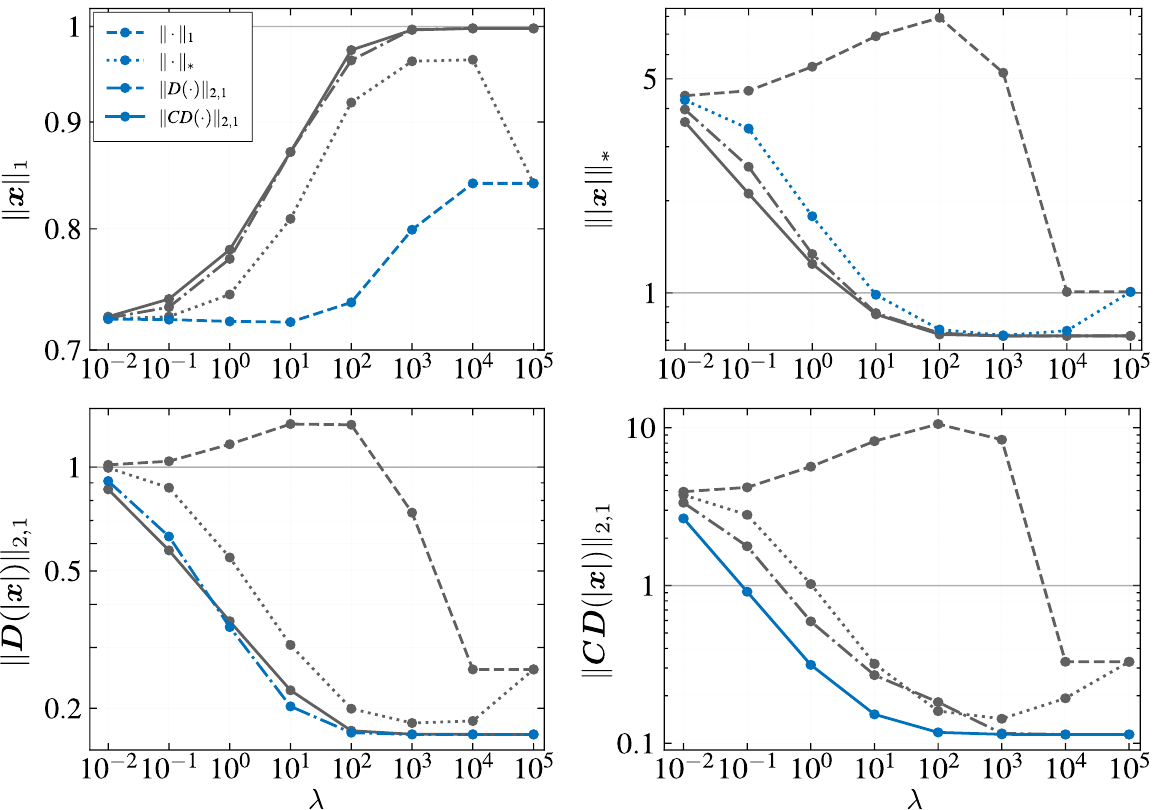}
    \vspace{-20pt}
    \caption{Values of $\psi(\boldsymbol{B}(|\boldsymbol{x}^{\star\!}|))/\psi(\boldsymbol{B}(|\boldsymbol{G}_w\boldsymbol{d}|))$.
    Each line corresponds to one of $\psi\circ\boldsymbol{B}$ in Example~\ref{example}.
    One of four lines is colored when $\psi\circ\boldsymbol{B}$ used for optimizing $\boldsymbol{\sigma}$ is the same as that used for the vertical axis.
    }
    \label{fig:penalty}
\end{figure}

Obtained T-F representations $|\boldsymbol{x}|$ and corresponding $\boldsymbol{\sigma}$ are shown in Fig.~\ref{fig:x_sigma_summary}.
Since $\|\boldsymbol{\sigma}\|_1$ is included in the definition of $\varphi(\boldsymbol{x},\boldsymbol{\sigma})$ in Eq.~\eqref{eq:sum_perspective}, $\lambda=0$ corresponds to basis pursuit in Eq.~\eqref{eq:basis_pursuit}.
By increasing $\lambda$, the structure induced by $\psi\circ\boldsymbol{B}$ is incorporated into the solution of basis pursuit.
However, due to the equality constraint, too large $\lambda$ enlarges difference between $|\boldsymbol{x}|$ and $\boldsymbol{\sigma}$, which distorts the effect of $\psi\circ\boldsymbol{B}$ on $|\boldsymbol{x}|$.
As in the figure, small and large $\lambda$ provide similar result, but some intermediate $\lambda$ gives distinctly different representations.

To quantitatively evaluate the difference, normalized penalty values for each representation were calculated as in Fig.~\ref{fig:penalty}.
As shown using the colored lines, by minimizing each penalty imposed on $\boldsymbol{\sigma}$, the same penalty was also minimized for $|\boldsymbol{x}|$.
Note that the norms $\|\cdot\|_1$ and $\|\cdot\|_\ast$, which induces $\boldsymbol{\sigma}=\boldsymbol{0}$ for huge $\lambda$, increased at the right end, whereas the seminorms $\|\boldsymbol{D}(\cdot)\|_{2,1}$ and $\|\boldsymbol{C}\boldsymbol{D}(\cdot)\|_{2,1}$, which induces constant $\boldsymbol{\sigma}$ for huge $\lambda$, stayed small.
Moreover, too large $\lambda$ seems to provide the same representation for different norms (or seminorms).
These results indicate that $\lambda$ interpolates between the solution of basis pursuit ($\lambda=0$) and some specific solution determined by the property of penalty function ($\lambda\to\infty$).

\begin{figure}[!t]
    \centering
    \includegraphics[scale=0.14]{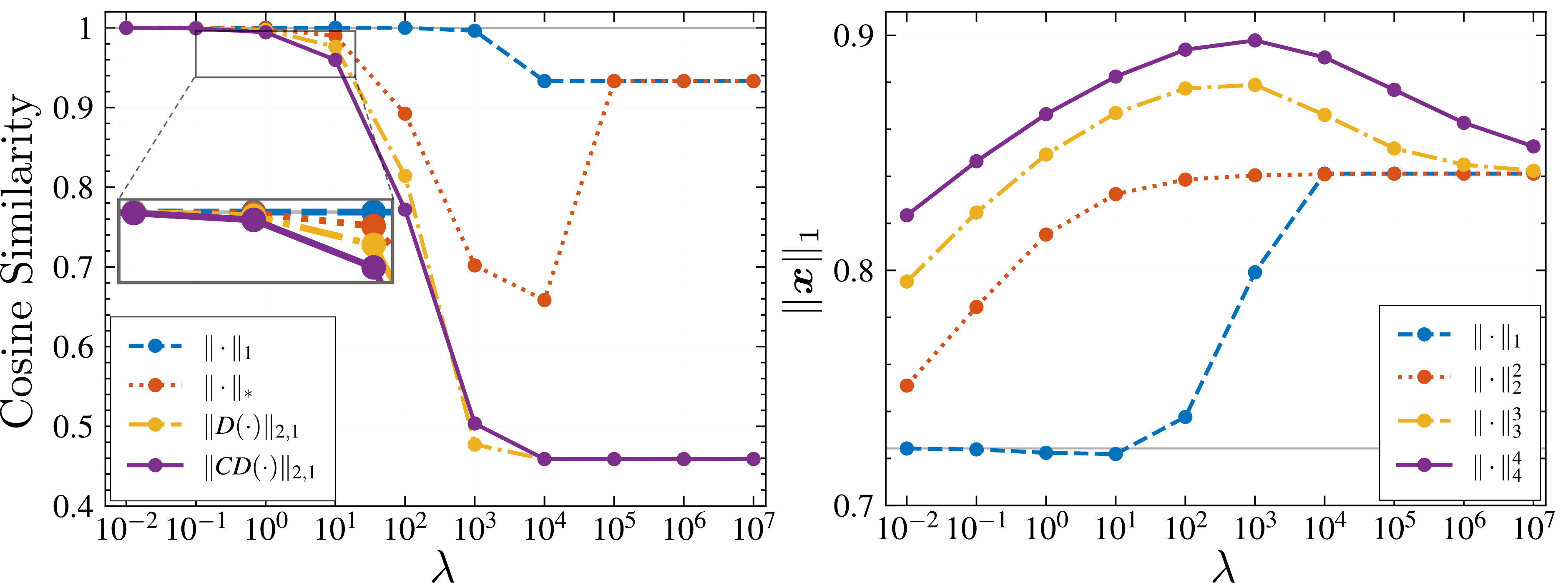}
    \vspace{-20pt}
    \caption{Cosine similarity between $|\boldsymbol{x}^{\star\!}|$ and $\boldsymbol{\sigma}$ (left), and $\|\boldsymbol{x}^{\star\!}\|_1/\|\boldsymbol{G}_w\boldsymbol{d}\|_1$ for $\boldsymbol{x}^{\star\!}$ obtained by $\psi = \|\cdot\|_p^p\:(p=1, 2, 3, 4)$ with $\boldsymbol{B} = \boldsymbol{I}$ (right).
    The solid horizontal line in the right figure indicates that obtained by $\lambda = 0$.}
    \label{fig:sim_sparse}
\end{figure}

Finally, we measured similarity between $|\boldsymbol{x}|$ and $\boldsymbol{\sigma}$ (left) and sparsity of $|\boldsymbol{x}|$ (right) as in Fig.~\ref{fig:sim_sparse}.
As in the left figure, $|\boldsymbol{x}|$ and $\boldsymbol{\sigma}$ are similar for small $\lambda$ but becomes different as $\lambda$ increases.
For larger $\lambda$, similarity converged to some values that are determined by the property of the penalty functions.
From the right figure, it can be seen that the obtained T-F representations were sparser than the DGT coefficient $\boldsymbol{G}_w\boldsymbol{d}$ (i.e., the normalized $\ell_1$-norm was less than $1$) even when the penalty function induces anti-sparsity ($p=3,4$).
Moreover, the starting point $(\lambda=0)$ and the end point $(\lambda\to\infty)$ seems the same for all norms.
Further investigation of these interesting properties of the proposed method is left as the future works.

\section{Conclusion}
In this paper, a convex optimization-based framework was proposed for realizing sparse T-F representations whose magnitude has properties specified by the user.
Some T-F representations that have not been realized before were provided to show the property of the proposed framework.
Future works include further investigation of the property of the proposed framework as well as the possible range of modification of the T-F representations.
Considering a denoising formulation by relaxing the equality constraint to inequality can be an interesting direction for extending the range of application.

\ifCLASSOPTIONcaptionsoff
  \newpage
\fi

\bibliographystyle{IEEEtran}
\bibliography{IEEEabrv,Bibliography}

\end{document}